\documentclass[italian,10pt,reqno]{amsart}

\usepackage{amsmath,amssymb}
\usepackage{latexsym}
\usepackage{epsfig}
\usepackage{cite}
\usepackage{booktabs}
\usepackage{color}
\usepackage{amsthm}
\newtheorem{prop}{\bf Proposition}[section]
\newtheorem{lemma}{\bf Lemma}[section]

\numberwithin{equation}{section}

\title{Path integral quantization of the relativistic Hopfield model}
\author{F. Belgiorno$^{1,2}$ \and S.L. Cacciatori$^{3,4}$ \and F. Dalla Piazza$^{5}$ \and M. Doronzo$^{3}$}
\address{\noindent $^1$Dipartimento di Matematica, Politecnico di Milano, Piazza Leonardo 32, IT-20133 Milano, Italy\endgraf
$^2$INdAM-GNFM \endgraf
$^3$Department of Science and High Technology, Universit\`a dell'Insubria, Via Valleggio 11, IT-22100 Como, Italy\endgraf
$^4$INFN sezione di Milano, via Celoria 16, IT-20133 Milano, Italy\endgraf
$^5$Universit\`a ``La Sapienza'', Dipartimento di Matematica, Piazzale A. Moro 2, I-00185, Roma, Italy}

\begin{document}
\maketitle
\begin{abstract}
The path integral quantization method is applied to 
a relativistically covariant version of the Hopfield model, 
which represents a very interesting mesoscopic framework for the description of the interaction between quantum light and dielectric quantum matter, with particular reference to the context of analogue gravity.  In order to take into account the constraints occurring in the 
model, we adopt the Faddeev-Jackiw approach to constrained quantization in the path integral formalism.
In particular we demonstrate that the propagator obtained with the Faddeev-Jackiw approach is equivalent to the one 
 which, in the framework of Dirac canonical quantization for constrained systems, can be directly computed as the vacuum expectation value of the time ordered product of the fields.   Our analysis also provides an explicit example of quantization of the electromagnetic field in a covariant gauge and coupled with the polarization field, which is a novel contribution to the literature on the Faddeev-Jackiw procedure.
\end{abstract}

\section{Introduction}
 In the context of the field represented by the interaction of the quantum electromagnetic field with quantum matter
 two different approaches can be adopted: the first one involves a microscopic description of the field and the second one considers a more phenomenological approach, in which some 
microscopic interactions are described by means of effective fields.
An example of this kind of approach is provided by models describing interactions of the 
electromagnetic field with dielectric media, 

which, beyond more standard applications to light-matter interactions, can be extended also to describe pair creation induced by an external field, by moving boundaries or by inhomogeneities propagating in the medium. Interest in this framework has been recently risen up, due to the attempt to 
reproduce quantum emission by black hole in the lab by means of analogous systems, i.e. systems displaying the same kinematics which is at the root of the 
Hawking effect \cite{hawking-cmp,unruh-seminal,barcelo,philbin}.\\ 
With respect to the phenomenological quantization of the electromagnetic field in the presence of a dielectric medium (as e.g. in \cite{luks}), the Hopfield model \cite{Hopfield1958} is able to describe the observed behaviour of the electromagnetic field in a class of transparent dielectric media by means of a very simple modeling of the matter itself, which is described as a collection of independent oscillators responsible for the dispersive properties of the electromagnetic field in matter \cite{fano,kittel,davydov}. 
These matter field degrees of freedom are represented by means of mesoscopic  polarization fields. Despite its simplicity, the model is able to reproduce 
the Sellmeier dispersion relations, which are fundamental features of light interactions in dielectric media. Still, in its original form the Hopfield model is not able to provide a description involving explicit relativistic covariance, which appears to be a fundamental request 
as far as one is interested in the analogue Hawking effect, as well as other perturbative and non perturbative effects in which an inhomogeneity propagates through a homogeneous background.
Indeed, the description of these phenomena requires the skill to move from an inertial frame to another one, for example from the lab frame to the frame comoving with the perturbation. 
A relativistic covariant version of the Hopfield model has been developed in \cite{PhysicaScripta}, together with its covariant and gauge invariant quantization. Therein, the quantization of a constrained system was taken into account by means of the Dirac quantization scheme, and states were constructed in the interaction picture. The Dirac approach for constrained systems is probably the most widely used \cite{dirac-book,dirac-papers,sundermeyer,Gitman-Tyutin,Henneaux-Teitelboim,Rothe}; it requires the identification and the classification of all the constraints of the theory into two classes, together with the redefinition of new brackets, the so called Dirac brackets. 
In \cite{EPJD2014} the quantization of the covariant Hopfield model was performed in the lab frame in a simple fixed gauge, in order to study photon production originated by time-dependent perturbations. While the analogue Hawking effect, beyond the analysis performed in 
\cite{NJP2011,petev,finazzi-carusotto-pra,finazzi-carusotto-pra14}, has been also analitically studied as a non perturbative effect in a simplified model in \cite{PRD2015}. The exact quantization of the model in absence of dielectric perturbations and 
in a generic covariant gauge, and the construction of the state in the Heisenberg representation have been implemented in \cite{Quantum}, where the mathematical issues and formal problems of the construction are discussed in detail.\\

An alternative and  powerful way to quantize a theory is provided by the path integral approach. In this paper we apply this quantization scheme to the relativistic covariant Hopfield model. Since the covariant Hopfield model is quadratic in the fields, it is essentially characterized by the propagator, which in the path integral quantization 
is obtained in a simpler manner than in the canonical quantization adopted in \cite{Quantum}. A difficulty of the model is given by the presence of constraints (such as the transversality condition for the  polarization  field), which impose 
 to adopt a quantization method for constrained theories. 
The extension of the path integral method to Dirac's theory of constrained systems was presented in \cite{Faddeev70}, 
with particular reference to first class constraints, and in \cite{Senjanovic76} 
with an analysis of  second class constraints. An alternative to the Dirac's method is represented by the Faddeev-Jackiw procedure \cite{Faddeev-Jackiw,Jackiw93}, which, with an elegant analysis, leads to the correct quantization without the necessity of the full Dirac machinery. 
In \cite{barcelos-neto,wotzasek,garcia,Liao} (see also \cite{Rothe}, section 4.4 therein) and, in the recent paper \cite{Toms} the implementation of the Faddeev-Jackiw method within the path integral approach is discussed.  We apply the Faddeev-Jackiw quantization scheme to the covariant version 
of the Hopfield model. Our aim is two-fold: on one hand, we find that 
the Faddeev-Jackiw procedure is simpler than the Dirac one. Indeed, applying this procedure to the Hopfield model, just one additional field appears in the quantization (see Section \ref{sec:path integral}), whereas three additional fields appear using the Dirac's method  (c.f. \cite{PhysicaScripta}, Section 4). Moreover, we obtain the same constraint, see Equation \eqref{cnstr}, as in the Dirac's procedure.
On the other hand, we also obtain an explicit example of quantization of the electromagnetic field in a covariant gauge, with a further 
difficulty represented by the coupling with the polarization field, which is missing in the previous literature on the   Faddeev-Jackiw procedure.\\
We point out that quantization of the original Hopfield model in non-covariant gauges and in presence of dissipation 
is performed in \cite{bechler1,bechler2}. In contrast, in our model there is no dissipation (which is a good 
approximation as far as phenomena occurring far from the absorptive regions are taken into account), but full 
relativistic covariance is accounted for.

The structure of the paper is the following. In Section 2.1 we will employ the Fadeev-Jackiw path integral procedure to quantize the covariant Hopfield model which possesses a singular Lagrangian; in Section 2.2 we will determine the exact propagator, discussing the right Feynman-St\"uckelberg prescription to be adopted; 
in Section 3 and in the appendices, we will demonstrate that the propagator obtained in this cleaner way is exactly equivalent to the one directly computed as the vacuum expectation value of the time ordered product of the fields \cite{Quantum}. \\ 
\\ 
We conclude with some comments on notations: we shall use the symbol $\pmb v$ for a space-time vector having components $v^\mu$ with $\mu=0,\ldots,3$, whereas its spatial component will be indicated with $\vec{v}$ or, more explicitly, $v^i$ with $i=1,2,3$. We shall use $v^2$ for the scalar $\pmb v^2 = \pmb v \cdot \pmb v$. Moreover, the Minkowski metric tensor is chosen with the signature $(+,-,-,-)$ and we take the speed of light $c=1$.

\section{Path integral formulation of the Hopfield model}
In place of using the standard Dirac method for quantizing a theory with constraints, we will implement the Faddeev-Jackiw method, which is based on recasting the Lagrangian in a first order formalism. We will refer to \cite{Toms}. 
 For the sake of completeness, we recall the essential ingredients characterizing the method, referring to the 
previously quoted literature for details. To a standard second order Lagrangian $L(q,\dot{q})$ we can associate a ``symplectic'' 
Lagrangian $L=\omega_\alpha (\xi) \dot{\xi}^\alpha -V(\xi)$ which is first order in time derivative $\dot{\xi}$, where $\xi=(q,p)$.  The first term in $L$ defines the so-called 
canonical one-form $\omega:=\omega_\alpha (\xi) d\xi^\alpha$. The second term $V$ can be identified with the Hamiltonian. 
The Euler-Lagrange equations become 
\begin{equation}
\Omega_{\alpha \beta}  \dot{\xi}^\beta = \frac{\partial V}{\partial \xi^\alpha},
\label{e-moto}
\end{equation}
where 
\begin{equation}
\Omega_{\alpha \beta}: = \frac{\partial}{\partial \xi^{\alpha}}\omega_{{\beta}}-\frac{\partial}
{\partial \xi^{\beta}}\omega_{{\alpha}},
\end{equation}
is the antisymmetric two-form associated with $L$. If $\det (\Omega)\not =0$, there is no constraint in the theory and 
quantization can proceed along the usual procedure. If instead $\det (\Omega)=0$, null  eigenvectors  $z_I$ of $\Omega$, such that 
\begin{equation}
z_I^\alpha \Omega_{\alpha \beta}= 0,
\end{equation}
are present, and their contraction with the equations of motion (\ref{e-moto}) 
\begin{equation}
\Lambda_I := z_I^\alpha  \frac{\partial V}{\partial \xi^\alpha}= 0,
\end{equation}
correspond to constraints $\Lambda_I =0$ of the theory. These constraints are then added to $L$ by introducing 
suitable Lagrange multipliers $y^I$, obtaining a new Lagrangian $L'$ in the extended canonical variables 
$(\xi^\alpha,y^I)$:
\begin{equation}
L':= L+y^I \Lambda_I. 
\end{equation}
A new canonical one-form $\omega'$ and a new two-form $\Omega'$ are obtained. If $\det (\Omega')\not =0$, the 
procedure ends; if not, the procedure is iterated till a nonsingular two-form is obtained by adding to the Lagrangian a suitable 
set of constraints obtained as above.\\
In the path-integral approach, the measure is chosen according to the prescription in \cite{Toms}.

\subsection{The Fadeev-Jackiw approach to path integral quantization}\label{sec:path integral}
In order to apply the Faddeev-Jackiw method, let us first note that the (auxiliary) field $B$ appears already at first order in the equations of motion, so we need only to introduce the momenta of the electromagnetic $\pmb A$ and of the polarization $\pmb P$ fields. Indeed,  the Lagrangian for the covariant relativistic Hopfield model for a single polarization field with resonance frequency $\omega_0$ is \cite{PhysicaScripta}:
\begin{align}
\mathcal{L}=&-\frac1{16\pi}F_{\mu \nu}F^{\mu \nu} - \frac1{2\chi \omega_0^2}[(v^{\rho}\partial_{\rho}P_{\mu})(v^{\rho}\partial_{\rho}P^{\mu})]+\frac1{2\chi}P_{\mu}P^{\mu} - \frac{g}{2}(v_{\mu}P_{\nu}-v_{\nu}P_{\mu})F^{\mu \nu} \cr
& + B(\partial_{\mu}A^{\mu}) + \frac{\xi}{2}B^2 + \lambda (v_{\mu}P^{\mu}),
\label{eq:Lc}
\end{align}
{where $\chi$ is the susceptibility, $g$ a coupling constant, $\xi$ the gauge fixing constant and $\lambda$ a lagrangian multiplier to impose the transversality condition for the polarization field.
In the first order formalism the Lagrangian reads:
\begin{align}
\label{lagfirstorder}
\mathcal{L}(\left\lbrace \pmb X \right\rbrace,\left\lbrace \pmb \Pi \right\rbrace)= \Pi_{\pmb A}^{\mu} \dot{A}_{\mu} + \Pi_{\pmb P}^{\mu} \dot{P}_{\mu} + \lambda (v_{\mu}P^{\mu}) - \mathcal{H},
\end{align}
with
\begin{align}
\mathcal{H}&= 2\pi \vec \Pi^2_{\pmb A} +\frac 1{16\pi} \sum_{i,j=1}^3 F_{ij}F^{ij} -A_0\sum_{i=1}^3 \partial_{x^i} \Pi^i_{\pmb A} -4\pi g \sum_{i=1}^3 (v_0 P_i-v_iP_0)\Pi_{\pmb A}^i-\sum_{j=1}^3 \frac {v^j}{v^0} \partial_{x^j} 
\pmb P\cdot  \pmb \Pi_{\pmb P} \\
& \quad -\frac {\chi_0 \omega_0^2}{2v_0^2} \pmb{\Pi_P} \cdot \pmb {\Pi_P}-\frac 1{2\chi_0} \pmb P\cdot \pmb P+2\pi g^2 (v_0 \vec P-\vec v P_0)^2+ g \sum_{i,j=1}^3 v_iP_jF^{ij} -B\sum_{i=1}^3 \partial_{x^i}A^i -\frac \xi2 B^2,
\end{align}
and $\left\lbrace \pmb X \right\rbrace$, $\left\lbrace \pmb \Pi \right\rbrace$ stand for the collection of the fields and the momenta respectively. 
We can slightly generalise the procedure explained in \cite{Toms} by treating the field $B$ apart, so that we introduce the canonical variables
\begin{align}
\xi^{\alpha}=\left( \Pi_{\pmb A}^{\mu};A_{\mu};\Pi_{\pmb P}^{\mu};P_{\mu};\lambda \right),
\end{align}
and the components of the canonical one-form $\omega$,  as can be read from \eqref{lagfirstorder}, are:
\begin{align}
\omega_{\Pi_{\pmb A}^{\mu}}&=0, \cr
\omega_{A_{\mu}}&=\Pi_{\pmb A}^{\mu}, \cr
\omega_{\Pi_{\pmb P}^{\mu}}&=0, \cr
\omega_{P_{\mu}}&=\Pi_{\pmb P}^{\mu}, \cr
\omega_{\lambda}&=v_{\mu}P^{\mu}.
\end{align}
If we choose $\xi^{\alpha}$ and $\xi^{\beta}$ with all the fields evaluated respectively at some spatial coordinate $\vec{x}$ and $\vec{x}'$, the symplectic two-form $\Omega_{\xi^{\alpha} \xi^{\beta}}=\frac{\delta}{\delta \xi^{\alpha}}\omega_{\xi^{\beta}}-\frac{\delta}
{\delta \xi^{\beta}}\omega_{\xi^{\alpha}}$ will take the form
\begin{align}
\Omega_{\xi^{\alpha} \xi^{\beta}}=
\left( 
\begin{matrix} 
0 & \delta_{\nu}^{\mu} & 0 & 0 & 0 \\ 
-\delta_{\mu}^{\nu} & 0 & 0 & 0 & 0 \\
0 & 0 & 0 & \delta_{\nu}^{\mu} & 0 \\ 
0 & 0 & -\delta_{\mu}^{\nu} & 0 & v^{\nu} \\
0 & 0 & 0 & -v^{\mu} & 0 \\
 \end{matrix} \right) \delta (\vec{x},\vec{x}').
\end{align}
Since the canonical variables are odd, we expect an odd number of missing constraints. Indeed, the first order Lagrangian is singular, i.e. $\det{(\Omega_{\xi^{\alpha} \xi^{\beta}})}=0$, with kernel of dimension one, generated by 
$z\equiv(0^{\mu};0_{\mu};v^\mu;0_{\mu};1)$. This mode is associated to a new constraint $\Lambda$: 
\begin{align}
\label{cnstr}
0=\Lambda=z^{\alpha}\frac{\delta}{\delta \xi^{\alpha}}\mathcal{H}=v^{\mu}\frac{\delta}{\delta \Pi_{\pmb P}^{\mu}}\mathcal{H}=\frac{v^{\mu}}{v_0}\left( v^k \partial_k P_{\mu} + \frac{\chi \omega_0^2}{v_0}\Pi_{\pmb P \mu} \right).
\end{align}
Inserting it into the Lagrangian in the first order formalism it yields 
\begin{align}
\mathcal{L}(\left\lbrace \pmb X \right\rbrace,\left\lbrace \pmb \Pi \right\rbrace)= \Pi_{\pmb A}^{\mu} \dot{A}_{\mu} + \Pi_{\pmb P}^{\mu} \dot{P}_{\mu} + \lambda (v_{\mu}P^{\mu}) +y\left[ \frac{v^{\mu}}{v_0}\left( v^k \partial_k P_{\mu} 
+ \frac{\chi \omega_0^2}{v_0}\Pi_{\pmb P \mu} \right) \right]  - \mathcal{H}.
\label{eq:L}
\end{align}
Thus, we must extend the canonical variables to
\begin{align}
\xi^{\alpha}=\left( \Pi_{\pmb A}^{\mu};A_{\mu};\Pi_{\pmb P}^{\mu};P_{\mu};\lambda;y \right),
\end{align}
and add the conjugate momentum
\begin{align}
\omega_y=\frac{v^{\mu}}{v_0}\left( v^k \partial_k P_{\mu} + \frac{\chi \omega_0^2}{v_0}\Pi_{\pmb P \mu} \right).
\end{align}
With the addition of the $y$ field, the symplectic two-form becomes
\begin{align}
\Omega_{\xi^{\alpha} \xi^{\beta}}=
\left( 
\begin{matrix} 
0 & \delta_{\nu}^{\mu} & 0 & 0 & 0 & 0\\ 
-\delta_{\nu}^{\mu} & 0 & 0 & 0 & 0 & 0\\
0 & 0 & 0 & \delta_{\nu}^{\mu} & 0 & \frac{\chi \omega_0^2}{v_0^2}v_{\nu}\\ 
0 & 0 & -\delta_{\nu}^{\mu} & 0 & v^{\nu} & \frac{v^{\nu}}{v_0}v^k \partial_k\\
0 & 0 & 0 & -v^{\mu} & 0 & 0\\
0 & 0 & -\frac{\chi \omega_0^2}{v_0^2}v_{\nu} & -\frac{1}{v_0}v^{\nu}v^k \partial_k & 0 & 0
 \end{matrix} \right) \delta (\vec{x},\vec{x}'),
\end{align}
which has non-singular determinant
\begin{align}
\left( \det{(\Omega_{\xi^{\alpha} \xi^{\beta}})} \right)^{1/2}=\det \left[ \frac{\chi \omega_0^2}{2 v_0^2} \delta (\vec{x},\vec{x}') \right].
\end{align}
Following Toms \cite{Toms}, we can proceed with the path integral approach and the path integral measure is:
\begin{align}
d\mu=\left( \prod_{\alpha} D\xi^{\alpha} \right) DB (\det{\Omega})^{1/2}.
\end{align}
Hence for the partition function we get
\begin{align}
Z_0= \pmb \int \left[ D \pmb A D \pmb {\Pi_A} D \pmb P D \pmb{\Pi_P} DB D\lambda Dy \right] (\det{\Omega})^{1/2} \exp{\left\lbrace i \int \mathcal{L}(\left\lbrace \pmb X \right\rbrace,\left\lbrace 
\pmb \Pi \right\rbrace)d^4\pmb x \right\rbrace },
\end{align}
with $\mathcal{L}$ given by (\ref{eq:L}). 
In order to recover the standard configuration space path integral we have to integrate over all momenta and the multiplicator fields. The integration over $\pmb{\Pi_A}$ is immediate, while integration\footnote{It's a Gaussian integral.} over $\pmb{\Pi_P}$ gives the contribution
\begin{align}
\pmb \int \prod_{\mu} \left[ D \Pi_{\pmb P}^{\mu} \right] \exp{\left\lbrace i \int \left[ \frac{\chi \omega_0^2}{2v_0^2}\Pi_{\pmb P \mu}\Pi_{\pmb P}^{\mu} + \left( \dot{P}_{\mu} + y \frac{\chi \omega_0^2}{v_0^2}v_{\mu} + 
\frac{v^k}{v_0}\partial_k P_{\mu} \right)\Pi_{\pmb P}^{\mu}  \right] d^4\pmb x 
\right\rbrace } = \left(\frac{\pi}{a}\right)^2 \exp{\left(-\frac{b_{\mu}b^{\mu}}{4a}\right)},
\end{align}
where $a=\frac{\chi \omega_0^2}{2v_0^2}$ and $b_{\mu}=\dot{P}_{\mu} + y \frac{\chi \omega_0^2}{v_0^2}v_{\mu} + \frac{v^k}{v_0}\partial_k P_{\mu}$. Similarly, integration over $y$ gives the contribution 
\begin{align}
\frac{v_0}{\omega_0}\sqrt{\frac{2\pi}{\chi}}\exp{\left(\frac{v_0^2}{2\chi \omega_0^2}(v_{\mu}\dot{P}^{\mu})^2\right)},
\end{align}
so that 
\begin{align}
Z_0=N \pmb \int \left[ D \pmb A D \pmb P DB D\lambda \right] \exp{\left\lbrace i \int \left[ \mathcal{L}(\pmb A,\pmb P,B,\lambda) + \frac{1}{2\chi \omega_0^2}(v_{\rho}v^0\partial_0 P^{\rho})(v_{\sigma}v^0\partial_0 
P^{\sigma}) \right]d^4\pmb x \right\rbrace },
\end{align}
where $N$ is a normalisation constant. Integration over $\lambda$ gives a factor $\delta(\pmb v \cdot \pmb P)$, so that finally we get
\begin{align}
Z_0=N  \pmb \int \left[ D \pmb A D \pmb P DB \right] \delta(\pmb v \cdot \pmb P) \exp{\left\lbrace i \int  \mathcal{L}(\pmb A,\pmb P,B) d^4\pmb x \right\rbrace }.
\end{align}
Obviously the normalisation must be such that $Z_0=1$.
\subsection{Determination of the propagator}
After introducing the currents $\pmb {J_A}, \pmb {J_P}$ and $J_B$, we can define the functional generating the propagators:
\begin{align}
Z[\pmb {J_A}, \pmb {J_P}, J_B]=\pmb \int [D\pmb A D \pmb P D B] &  \delta(\pmb v \cdot \pmb P) \exp \left\{ i \int_{\mathbb R^4} \mathcal L(\pmb A,\pmb P,B) d^4 \pmb x + i \int_{\mathbb R^4} J_{\pmb A}^\mu A_\mu d^4 \pmb x
+ i \int_{\mathbb R^4} J_{\pmb P}^\mu P_\mu d^4 \pmb x+ \right. \cr
& \qquad\ \left.  i \int_{\mathbb R^4} J_{B} B d^4 \pmb x 
\right\}.
\end{align}
In order to compute it, we rewrite the delta function in terms of the integration over the field $\lambda$, by restoring the Lagrangian $\mathcal{L} \equiv \mathcal{L}(\pmb A,\pmb P,B,\lambda)$ of (\ref{eq:Lc}). Moreover, we introduce the current $J_\lambda$,
which will be set to zero at the end of the calculation, in order to simplify some technical step. Thus, we consider the generating functional
\begin{align}
Z[\pmb {J_A}, \pmb {J_P}, J_B, J_\lambda]=\pmb \int [D\pmb A D \pmb P D B D \lambda] & \exp \left\{ i \int_{\mathbb R^4} \mathcal{L}(\pmb A,\pmb P,B,\lambda) d^4 \pmb x + i \int_{\mathbb R^4} J_{\pmb A}^\mu A_\mu d^4 \pmb x
+ i \int_{\mathbb R^4} J_{\pmb P}^\mu P_\mu d^4 \pmb x+ \right. \cr
& \qquad\ \left.  i \int_{\mathbb R^4} J_{B} B d^4 \pmb x + i \int_{\mathbb R^4} J_{\lambda} \lambda d^4 \pmb x
\right\}.
\end{align}
After passing to the Fourier transforms of the fields and the currents, we proceed in the usual way
\begin{align}
Z[\pmb {J_A}, \pmb {J_P}, J_B, J_\lambda]=\pmb \int [D\pmb {\mathcal A} D \pmb {\mathcal P} D \mathcal B D \tilde \lambda] & \exp \left\{ -\frac i2 \int_{\mathbb R^4} \frac {d^4\pmb k}{(2\pi)^4} \tilde{\pmb V}(-\pmb k) 
{\mathcal M}_{i\varepsilon} (\pmb k) \tilde {\pmb V}(\pmb k) \right. \cr
& \left.+ \frac i2 \int_{\mathbb R^4} \frac {d^4 \pmb k}{(2\pi)^4} \left[ \tilde {\pmb J}_{\pmb V} (-\pmb k)^T \tilde {\pmb V}(\pmb k) +  \tilde {\pmb V}(-\pmb k)^T \tilde {\pmb J}_{\pmb V} (\pmb k) \right]   
\right\}, \label{Z}
\end{align}
where 
\begin{eqnarray}
\tilde {\pmb V}=\begin{pmatrix}
\pmb {\mathcal A} \\ \pmb {\mathcal P} \\ \mathcal B\\ \tilde \lambda
\end{pmatrix},
\end{eqnarray}
and
\begin{eqnarray}
\tilde {\pmb J}_{\pmb V}=
\begin{pmatrix}
\tilde {\pmb J}_{\pmb A} \\ \tilde {\pmb J}_{\pmb P}\\ \tilde J_B\\ \tilde J_\lambda
\end{pmatrix},
\end{eqnarray}
are the Fourier transform of the fields and the currents respectively, and we have introduced the Feynman-St\"uckelberg prescription, to be correctly identified, associated to the Fourier space operator
\begin{align}
{\mathcal M} \tilde {\pmb V} \equiv 
\begin{pmatrix}
\frac 1{4\pi} (k^2\mathbb I-\pmb k \pmb k^t) & ig (\omega \mathbb I -\pmb v \pmb k^t ) & -i\pmb k & \pmb 0 \\
-ig (\omega \mathbb I -\pmb k \pmb v^t ) & \frac 1{\chi} \left( \frac {\omega^2}{\omega_0^2}-1 \right)\mathbb I & \pmb 0 & \pmb v \\
i\pmb k^t & \pmb 0^t & -\xi & 0 \\
\pmb 0^t & \pmb v^t & 0 & 0
\end{pmatrix}
\begin{pmatrix}
\pmb{\mathcal A} \\ \pmb {\mathcal P} \\ \mathcal B\\ \tilde \lambda
\end{pmatrix}
=
\begin{pmatrix}
0 \\ 0 \\ 0\\ 0
\end{pmatrix},
\end{align}
where $\omega:=k^\mu v_\mu$, $\pmb k$ is a column vector, $\pmb k^t$ a row vector, and so on. The determinant of this matrix is 
\begin{eqnarray}
\det  {\mathcal M}=-\frac {(k^2)^2}{4\pi \omega_0^6\chi^3} \left( \omega^2-\omega_0^2 \right)^2 \left[\frac {k^2}{4\pi} -\frac {g^2\chi \omega_0^2 \omega^2}{\omega^2-\omega_0^2} \right]^2
(\omega^2-\bar\omega^2),
\end{eqnarray}
so that $\mathcal M^{-1}$ is singular at $k^2=0$, and\footnote{Note that the determinant does not vanish when $\omega^2=\omega_0^2$.}
\begin{eqnarray}
\frac {k^2}{4\pi} -\frac {g^2\chi_0 \omega_0^2 \omega^2}{\omega^2-\omega_0^2}=0, \qquad\ \omega^2-\bar\omega^2=0,
\end{eqnarray}
thus requiring a prescription in order to avoid singularities on the real spectrum. Let us postpone momentarily its specification. Then, we can proceed as usual by shifting the fields 
\begin{align}
\tilde {\pmb V}(\pmb k)={\mathcal M}_{i\varepsilon}^{-1} (\pmb k) \tilde {\pmb J}_{\pmb V}(\pmb k)+\tilde {\pmb \Phi}(\pmb k),
\end{align}
where 
\begin{eqnarray}
\tilde {\pmb \Phi}\equiv \begin{pmatrix}
\pmb {a} \\ \pmb { p} \\ b\\ \tilde \phi
\end{pmatrix},
\end{eqnarray}
so that
\begin{align}
Z[\pmb {J_A}, \pmb {J_P}, J_B, J_\lambda]=&
\exp \left\{ \frac i2 \int_{\mathbb R^4} \frac {d^4 \pmb k}{(2\pi)^4} e^{-i\pmb k \cdot (\pmb x-\pmb y)} \tilde {\pmb J}_{\pmb V}(-\pmb k)^T  {\mathcal M}_{i\varepsilon}^{-1} (\pmb k)  \tilde {\pmb J}_{\pmb V}(\pmb k)  \right\}\cdot \cr
&\cdot\pmb \int [D\pmb {a} D \pmb {p} D b D \tilde \phi] \exp \left\{ -\frac i2 \int_{\mathbb R^4} \frac {d^4\pmb k}{(2\pi)^4} \tilde{\pmb \Phi}(-\pmb k) 
{\mathcal M}_{i\varepsilon} (\pmb k) \tilde {\pmb \Phi}(\pmb k) \right\}
\end{align}
and, since
\begin{eqnarray}
Z[\pmb 0, \pmb 0, 0, 0]=Z_0= \langle 0|0 \rangle =1
\end{eqnarray}
after setting $J_\lambda=0$, we finally get
\begin{align}
Z[\pmb {J_A}, \pmb {J_P}, J_B]=\exp \left\{ \frac i2 \int_{\mathbb R^4 \times \mathbb R^4} d^4 \pmb x d^4 \pmb y \pmb J^T (\pmb x) \pmb {G}_F(\pmb x-\pmb y) \pmb J(\pmb y) \right\},
\end{align}
where
\begin{align}
\pmb J(\pmb x)=
\begin{pmatrix}
\pmb {J_A}\\ \pmb {J_P}\\ J_B
\end{pmatrix},
\end{align}
and
\begin{eqnarray}
\pmb {G}_F(\pmb x)= \int_{\mathbb R^4} \frac {d^4 \pmb k}{(2\pi)^4} e^{-i\pmb k\cdot \pmb x} \bar {{\mathcal M}}_{i\varepsilon}^{-1} (\pmb k).
\end{eqnarray}
Here, with $\bar {{\mathcal M}}_{i\varepsilon}^{-1} (\pmb k)$ we mean the reduction of the matrix ${\mathcal M}_{i\varepsilon}^{-1} (\pmb k)$ after dropping the last row and last column out, as specified below in (\ref{tilde M-1 reduc}).\\
Now, we specify the Feynman-St\"uckelberg prescription (more on this can be found in Appendix \ref{app:propagator}). We define it by means of the complex shifts $k^2 \to k^2+i\varepsilon$, $\omega_0^2\to \omega_0^2-i\varepsilon$, t.i.:
\begin{align}
{\mathcal M}_{i\varepsilon} (\pmb k):= 
\begin{pmatrix}
\frac 1{4\pi} ((k^2+i\varepsilon)\mathbb I-\pmb k \pmb k^t) & ig (\omega \mathbb I -\pmb v \pmb k^t ) & -i\pmb k & \pmb 0 \\
-ig (\omega \mathbb I -\pmb k \pmb v^t ) & \frac 1{\chi} \left( \frac {\omega^2}{\omega_0^2-i\varepsilon}-1 \right)\mathbb I & \pmb 0 & \pmb v \\
i\pmb k^t & \pmb 0^t & -\xi & 0 \\
\pmb 0^t & \pmb v^t & 0 & 0
\end{pmatrix}
\end{align}
from which we get 
\begin{eqnarray}
\mathcal M^{-1}_{i\varepsilon}(\pmb k)=
\begin{pmatrix}
M^{-1}_{i\varepsilon} & N_{i\varepsilon} & -i\frac 1{k^2+{i\varepsilon}}\pmb k & \pmb 0 \\
N^\dagger_{i\varepsilon} & Q_{i\varepsilon}+\frac {\chi (\omega_0^2-{i\varepsilon})}{\omega^2-\omega_0^2+{i\varepsilon}} \left(\mathbb I- \pmb v \pmb v^t \right)& \pmb 0  & \pmb v \\
i\frac 1{k^2+{i\varepsilon}} \pmb k^t & \pmb 0^t & 0 & 0 \\
\pmb 0^t & \pmb v^t & 0 & - \frac {\omega^2-\omega_0^2+{i\varepsilon}}{\chi (\omega_0^2-{i\varepsilon})}
\end{pmatrix}. \label{tilde M-1}
\end{eqnarray}
where
\begin{eqnarray}
&&M_{i\varepsilon}^{-1}=\frac 1{\frac {k^2+i\varepsilon}{4\pi} -\frac {g^2\chi \omega_0^2 \omega^2}{\omega^2-\omega_0^2}}\mathbb I+\left[ \frac 1{(k^2+i\varepsilon)^2} \left( \xi-\frac 1{\frac 1{4\pi} 
-\frac {g^2\chi (\omega_0^2-i\varepsilon)}{\omega^2-\omega_0^2+i\varepsilon}} \right)\right.\cr
&&\phantom{M^{-1}=}\left.+\frac 1{\frac {k^2+i\varepsilon}{4\pi} -\frac {g^2\chi (\omega_0^2-i\varepsilon) \omega^2}{\omega^2-\omega_0^2+i\varepsilon}} \frac 1{k^2+i\varepsilon} \frac 1{\frac 1{4\pi} 
-\frac {g^2\chi (\omega_0^2-i\varepsilon)}{\omega^2-\omega_0^2+i\varepsilon}}\frac {g^2\chi (\omega_0^2-i\varepsilon)}{\omega^2-\omega_0^2+i\varepsilon} \right]
\pmb k \pmb k^t\cr
&&\phantom{M^{-1}=}-\frac {4\pi}{\frac {k^2+i\varepsilon}{4\pi} -\frac {g^2\chi (\omega_0^2-i\varepsilon) \omega^2}{\omega^2-\omega_0^2+i\varepsilon}} 
\frac {\omega^2-\omega_0^2+i\varepsilon}{\omega^2-\bar\omega^2+i\varepsilon}
 \frac {g^2\chi (\omega_0^2-i\varepsilon)}{\omega^2-\omega_0^2+i\varepsilon}
\left[\omega\frac { (\pmb k \pmb v^t+\pmb v \pmb k^t)}{k^2+i\varepsilon}- \pmb v \pmb v^t \right] \label{Mdritta-1}
\end{eqnarray}
is the inverse of the matrix
\begin{align}
M_{i\varepsilon}&=\left(\frac {k^2+i\varepsilon}{4\pi} -\frac {g^2\chi (\omega_0^2-i\varepsilon) \omega^2}{\omega^2-\omega_0^2+i\varepsilon}\right)\mathbb I-\frac 1{4\pi} \left( 1-\frac {4\pi}\xi \right) \pmb k \pmb k^t +\omega 
\frac {g^2\chi (\omega_0^2-i\varepsilon)}{\omega^2-\omega_0^2+i\varepsilon}(\pmb k \pmb v^t+\pmb v \pmb k^t)\cr 
&-(k^2+i\varepsilon)\frac {g^2\chi (\omega_0^2-i\varepsilon)}{\omega^2-\omega_0^2+i\varepsilon} \pmb v \pmb v^t, \label{MM}
\end{align}
whereas
\begin{eqnarray}
&& N_{i\varepsilon}=-ig \frac {\chi (\omega_0^2-i\varepsilon)}{\omega^2-\omega_0^2+i\varepsilon} M^{-1}(\omega \mathbb I-\pmb v \pmb k^t)\cr
&& \phantom{N_{i\varepsilon}}=-ig \frac {\chi (\omega_0^2-i\varepsilon)}{\omega^2-\omega_0^2+i\varepsilon}\frac 1{\frac {k^2+i\varepsilon}{4\pi} -\frac {g^2\chi (\omega_0^2-i\varepsilon) \omega^2}{\omega^2-\omega_0^2+i\varepsilon}}
\left[ \omega \mathbb I+\frac 1{\frac 1{4\pi} -\frac {g^2\chi (\omega_0^2-\varepsilon)}{\omega^2-\omega_0^2+i\varepsilon}} \frac {\omega g^2 \chi (\omega_0^2-i\varepsilon)}{\omega^2-\omega_0^2+i\varepsilon}
\left( \pmb v \pmb v^t \right. \right. \cr 
&& \phantom{N_{i\varepsilon}=} +\frac 1{k^2+i\varepsilon} \pmb k \pmb k^t
\left.\left. -\frac \omega{ k^2+i\varepsilon} \pmb v \pmb k^t-\frac {\omega^2-\omega_0^2+i\varepsilon}{4\pi \omega g^2 \chi (\omega_0^2-i\varepsilon)}\pmb k \pmb v^t \right) \right], \label{N}\\
&& N_{i\varepsilon}^\dagger=i g \frac {\chi (\omega_0^2-i\varepsilon)}{\omega^2-\omega_0^2+i\varepsilon} (\omega \mathbb I-\pmb v \pmb k^t)M^{-1}\cr
&& \phantom{N_{i\varepsilon}^\dagger}=i g \frac {\chi (\omega_0^2-i\varepsilon)}{\omega^2-\omega_0^2+i\varepsilon}\frac 1{\frac {k^2+i\varepsilon}{4\pi} -\frac {g^2\chi (\omega_0^2-i\varepsilon) \omega^2}{\omega^2-\omega_0^2+i\varepsilon}}
\left[ \omega \mathbb I+\frac 1{\frac 1{4\pi} -\frac {g^2\chi (\omega_0^2-\varepsilon)}{\omega^2-\omega_0^2+i\varepsilon}} \frac {\omega g^2 \chi (\omega_0^2-i\varepsilon)}{\omega^2-\omega_0^2+i\varepsilon}
\left( \pmb v \pmb v^t \right. \right. \cr 
&& \phantom{N_{i\varepsilon}^\dagger=} +\frac 1{k^2+i\varepsilon} \pmb k \pmb k^t
\left.\left. -\frac \omega{k^2+i\varepsilon} \pmb v \pmb k^t-\frac {\omega^2-\omega_0^2+i\varepsilon}{4\pi \omega g^2 \chi (\omega_0^2-i\varepsilon)}\pmb k \pmb v^t \right) \right].\label{N-daga}
\end{eqnarray}
Notice that $N_{i\varepsilon}^\dagger$ is the Hermitian conjugate of $N_{i\varepsilon}$ only when $\varepsilon=0$.
Finally
\begin{eqnarray}
&& Q_{i\varepsilon}=g^2 \frac {\chi^2 (\omega_0^2-i\varepsilon)^2}{(\omega^2-\omega_0^2+i\varepsilon)^2} (\omega \mathbb I-\pmb k \pmb v^t)M^{-1} (\omega \mathbb I-\pmb v \pmb k^t)\cr
&& \phantom{Q_{i\varepsilon}}=g^2 \frac {\chi^2 (\omega_0^2-i\varepsilon)^2}{(\omega^2-\omega_0^2+i\varepsilon)^2} \frac 1{\frac {k^2+i\varepsilon}{4\pi} 
-\frac {g^2\chi (\omega_0^2-i\varepsilon) \omega^2}{\omega^2-\omega_0^2+i\varepsilon}}
\left[ \omega^2\mathbb I \phantom{\frac 1{\frac 1{4\pi} -\frac {g^2\chi (\omega_0^2-i\varepsilon) }{\omega^2-\omega_0^2+i\varepsilon}}} \right. \cr
&& \phantom{Q_{i\varepsilon}=}\left.
+\frac 1{\frac 1{4\pi} -\frac {g^2\chi (\omega_0^2-i\varepsilon) }{\omega^2-\omega_0^2+i\varepsilon}}\left(
\frac {g^2\chi (\omega_0^2-i\varepsilon) \omega^2}{\omega^2-\omega_0^2+i\varepsilon} \pmb v \pmb v^t -\frac \omega{4\pi} (\pmb k \pmb v^t+\pmb v \pmb k^t)+ \frac 1{4\pi} \pmb k \pmb k^t \right) \right]. \label{Q}
\end{eqnarray}
A deduction of these formulas is presented in Appendix \ref{sec:M-1}.\\
In conclusion
\begin{eqnarray}
\bar {\mathcal M}^{-1}_{i\varepsilon}(\pmb k)=
\begin{pmatrix}
M^{-1}_{i\varepsilon} & N_{i\varepsilon} & -i\frac 1{k^2+{i\varepsilon}}\pmb k  \\
N^\dagger_{i\varepsilon} & Q_{i\varepsilon}+\frac {\chi (\omega_0^2-{i\varepsilon})}{\omega^2-\omega_0^2+{i\varepsilon}} \left(\mathbb I- \pmb v  \pmb v^t \right)& \pmb 0  \\
i\frac 1{k^2+{i\varepsilon}} \pmb k^t & \pmb 0^t & 0 
\end{pmatrix}. \label{tilde M-1 reduc}
\end{eqnarray}
\subsubsection{ Remark:} A comment is in order. The consistency of the constraints would require to impose the condition $\pmb v\cdot \pmb J_{\pmb P}=0$ also. However, because of the condition $\pmb v\cdot \pmb P=0$, we see that
after leaving $\pmb J_{\pmb P}$ unconstrained we have that 
\begin{eqnarray}
v^\mu \frac {\delta}{\delta J_P^\mu(\pmb x)} Z[\pmb {J_A}, \pmb {J_P}, J_B]=0,
\end{eqnarray}
so we don't need to take care of the constraint. This is consistent with the fact that the vector 
\begin{eqnarray}
\pmb V_0 =\begin{pmatrix}
\pmb 0 \\ \pmb v \\ 0
\end{pmatrix},
\end{eqnarray}
is in both the left kernel and the right kernel of (\ref{tilde M-1 reduc}).
\subsubsection{ Remark:} Exactly the same result can be obtained by using the Dirac procedure, even though in that case one has to introduce a larger number of auxiliary fields, see \cite{PhysicaScripta}. Another way to apply the Faddeev-Jackiw
method is to generate a canonical momentum for the field $B$ by adding a kinematical therm for it. In this case the Lagrangian becomes
\begin{align}
\mathcal{L}=&-\frac1{16\pi}F_{\mu \nu}F^{\mu \nu} - \frac1{2\chi \omega_0^2}[(v^{\rho}\partial_{\rho}P_{\mu})(v^{\rho}\partial_{\rho}P^{\mu})]+\frac1{2\chi}P_{\mu}P^{\mu} - \frac{g}{2}(v_{\mu}P_{\nu}-v_{\nu}P_{\mu})F^{\mu \nu} \cr
& + B(\partial_{\mu}A^{\mu}) + \frac{\xi}{2}B^2 - \frac{\bar{\xi}}{2}(\partial_{\mu}B)(\partial^{\mu}B) + \lambda (v_{\mu}P^{\mu}),
\end{align}
and the momentum conjugate to $B$ is $\Pi_B=-\frac{\bar{\xi}}{c} \partial_0 B$. In this way the field $B$ can be included exactly at the same footing as the other fields, and proceeding as above one finally gets the same result where now in the momentum space $\xi$ is replaced by $\xi+\bar \xi k^2$. However, the price would be to introduce a new parameter, $\bar \xi$, which is expected to be zero, since we have not vacuum polarization (see \cite{Gitman-Tyutin}). Moreover,
$B$ appears at first order in the equations of motion, already in the starting problem. Thus, there are no reasons for promoting it to the second order with the aim of going back to the first order formalism.

\section{The propagator}
The exact propagator $\pmb G(\pmb x, \pmb y)$ of the relativistic Hopfield model has been computed in \cite{Quantum} by using the oscillator representation. For convenience we report here the result. It can be written as
\begin{eqnarray}
iG^{IJ}(\pmb x,\pmb y)=\langle 0| T(\Phi^I(\pmb x)\Phi^J(\pmb y))|0\rangle, \qquad I,J=1,\ldots,9,
\end{eqnarray}
where 
\begin{eqnarray}
\Phi^I=
\begin{cases}
A^{I-1} & \mbox{if $I=1,2,3,4$}, \\
P^{I-5} & \mbox{if $I=5,6,7,8$}, \\
B & \mbox{if $I=9$},
\end{cases}
\end{eqnarray}
and takes the form
\begin{eqnarray}
G^{IJ}(\pmb x,\pmb y)=G^{IJ}(\pmb x,\pmb y)_+\theta (x^0-y^0)+G^{IJ}(\pmb x,\pmb y)_-\theta(y^0-x^0),
\end{eqnarray}
where
\begin{align}
iG^{(\mu+1)(\nu+1)}(\pmb x,\pmb y)_+=&\langle 0| A^\mu(\pmb x)A^\nu(\pmb y)|0\rangle\cr
=&\int_{\mathbb R^3} \frac {d^3\vec k}{(2\pi)^3}e^{-i\pmb k_+\cdot (\pmb x-\pmb y)}
\frac {v^\mu k_+^\nu+v^\nu k_+^\mu}{|\vec k| \omega_+} \pi \left(\frac \xi{4\pi}+\frac  {\omega_+^2-\omega_0^2}{\omega_+^2-\bar\omega^2}  \right)\cr
&-i \int_{\mathbb R^3} \frac {d^3\vec k}{(2\pi)^3}e^{-i\pmb k_+\cdot (\pmb x-\pmb y)}\frac {(\pmb x-\pmb y)\cdot \pmb v}{|\vec k|\omega_+}\left(\frac \xi{4\pi}-\frac  {\omega_+^2-\omega_0^2}{\omega_+^2-\bar\omega^2} \right) \pi k^\mu k^\nu\cr
&+ \int_{\mathbb R^3} \frac {d^3\vec k}{(2\pi)^3}e^{-i\pmb k_+\cdot (\pmb x-\pmb y)} \frac {\frac \xi{4\pi}(\omega_+^2-\bar\omega^2)^2+(\omega_+^2-\omega_0^2)
(\omega_+^2-\bar\omega^2) -8\pi\omega_+^2 g^2 \chi \omega_0^2}{\omega_+^3 |\vec k| (\omega_+^2-\bar \omega^2)^2}
\pi k^\mu k^\nu\cr
&+\frac {8\pi^2 g^2 \chi \omega_0^2}{\bar \omega v^0} \int_{\mathbb R^3} \frac {d^3\vec k}{(2\pi)^3}e^{-i\pmb k_>\cdot (\pmb x-\pmb y)}\frac 1{\bar\omega^2- k^2_>} \left( v^\mu-\frac {\bar\omega}{k^2_>} k^\mu_>\right)
\left( v^\nu-\frac {\bar\omega}{k^2_>} k^\nu_>\right)\cr
&+\sum_{a=1}^2 \sum_{i=1}^2 \int_{\mathbb R^3} \frac {d^3\vec k}{(2\pi)^3}e^{-i\pmb k_{(a)}\cdot (\pmb x-\pmb y)} \frac {e^{(a)\mu}_i(\vec k) e^{(a)\nu}_i (\vec k)}{DR'_{(a)}(\vec k)}; \label{AA}
\end{align}
\begin{align}
iG^{(\mu+1)(\nu+5)}(\pmb x,\pmb y)_+=&\langle 0| A^\mu(\pmb x)P^\nu(\pmb y)|0\rangle\cr
=&-i \int_{\mathbb R^3} \frac {d^3\vec k}{(2\pi)^3}e^{-i\pmb k_+\cdot (\pmb x-\pmb y)} \frac {2\pi g\chi \omega_0^2}{|\vec k| \omega_+} \frac {\omega_+ k_+^\mu v^\nu- k_+^\mu k_+^\nu}{(\omega_+^2-\bar\omega^2)} \cr
&-i\frac {2\pi g\omega_0^2 \chi}{\bar \omega v^0} \int_{\mathbb R^3} \frac {d^3\vec k}{(2\pi)^3}e^{-i\pmb k_>\cdot (\pmb x-\pmb y)} \frac 1{\bar \omega^2-k_>^2} \left( v^\mu-\frac {\bar \omega}{k^2_>} k^\mu_> \right) (\bar \omega v^\nu-k_>^\nu)\cr
& -ig \chi \omega_0^2  \sum_{a=1}^2 \sum_{i=1}^2 \int_{\mathbb R^3} \frac {d^3\vec k}{(2\pi)^3}e^{-i\pmb k_{(a)}\cdot (\pmb x-\pmb y)} \frac {\omega_{(a)}}{\omega_{(a)}^2-\omega_0^2} 
\frac {e^{(a)\mu}_i(\vec k) e^{(a)\nu}_i (\vec k)}{DR'_{(a)}(\vec k)}; \label{AP}
\end{align}
\begin{align}
iG^{(\mu+1)9}(\pmb x,\pmb y)_+=&\langle 0| A^\mu(\pmb x)B(\pmb y)|0\rangle=-\frac i2 \int_{\mathbb R^3} \frac {d^3\vec k}{(2\pi)^3}e^{-i\pmb k_+\cdot (\pmb x-\pmb y)} \frac {k^\mu}{|\vec k|}; \label{AB}
\end{align}
\begin{align}
iG^{(\mu+5)(\nu+5)}(\pmb x,\pmb y)_+=&\langle 0| P^\mu(\pmb x)P^\nu(\pmb y)|0\rangle\cr
=&\frac {\chi \omega_0^2}{2 \bar\omega v^0} \int_{\mathbb R^3} \frac {d^3\vec k}{(2\pi)^3}e^{-i\pmb k_>\cdot (\pmb x-\pmb y)} \frac 1{\bar \omega^2-\omega_0^2} (\bar \omega v^\mu- k_>^\mu)(\bar \omega v^\nu- k_>^\nu)\cr
& +g^2 \chi^2 \omega_0^4  \sum_{a=1}^2 \sum_{i=1}^2 \int_{\mathbb R^3} \frac {d^3\vec k}{(2\pi)^3}e^{-i\pmb k_{(a)}\cdot (\pmb x-\pmb y)} \frac {\omega_{(a)}}{\omega_{(a)}^2-\omega_0^2} 
\frac {e^{(a)\mu}_i(\vec k) e^{(a)\nu}_i (\vec k)}{DR'_{(a)}(\vec k)}; \label{PP}
\end{align}
\begin{align}
iG^{(\mu+5)9}(\pmb x,\pmb y)_+=&\langle 0| P^\mu(\pmb x)B(\pmb y)|0\rangle=0; \label{PB}
\end{align}
\begin{align}
iG^{99}(\pmb x,\pmb y)_+=&\langle 0| B(\pmb x)B(\pmb y)|0\rangle=0. \label{BB}
\end{align}
We refer to \cite{Quantum} for the notations.
We want to compare this expression for the propagator with the results of the previous section. Our main result is
\begin{prop}\label{proposition}
The propagator is
\begin{eqnarray}
\pmb G(\pmb x, \pmb y)= \pmb G_F(\pmb x-\pmb y), \label{propagatore}
\end{eqnarray}
\end{prop}
The proof of this proposition is given in Appendix \ref{app:propagator}.

\section{Conclusions}

In this paper we  have dealt with the quantization of the relativistic covariant Hopfield model via the path integral approach. 
As our model is Gaussian, it is completely determined by the two point function, i.e. the propagator, which can be computed either in a canonical quantization approach as the two point function of the fields 
(as we did in \cite{Quantum}, with considerable efforts), or in a relatively  straightforward way in the path integral formulation. 
Due to the presence of a constraint in the theory, the path integral implementation of the Faddeev-Jackiw method for constrained theories has been used. 
In contrast to the more standard Dirac's method,
which was adopted in \cite{PhysicaScripta,Quantum}, the Faddev-Jackiw approach is simpler and avoids the division of the constraints into different classes and the redefinition of the Poisson brakets. This represents our first interesting contribution, as we provide a non-trivial example for 
the electromagnetic field quantization in a covariant gauge in the Faddeev-Jackiw framework. \\
In particular, starting from the singular Lagrangian (\ref{eq:Lc}), we have obtained a new constraint, $\Lambda$, which is identical to the one  which emerged as a second-class second-stage constraint in the Dirac's procedure \cite{PhysicaScripta}. 
Computing the functional measure for the path integral has been then straightforward, as also pointed out by Toms \cite{Toms}, and the standard procedure for determining the quantum field theory partition function in terms of the propagator 
 has been implemented. 
An important question to be taken into account has been  the right choice for the Feynman-St\"uckelberg prescription, as pointed out in Appendix \ref{app:propagator}, which must ensure that positive square norm states propagate forward in time. \\
A further key-result of our analysis consists in the equivalence between the exact propagator obtained from the direct calculation, coming from the canonical quantization formalism,  and the one coming from the Faddeev-Jackiw path integral quantization.



\newpage
\begin{appendix}

\section{Computation of $\mathcal M^{-1}$}\label{sec:M-1}
For simplicity, we will compute the matrix $M^{-1}$ for $\varepsilon=0$, from which the computation of $\mathcal M^{-1}_{i\varepsilon}$ is obvious. For this, we apply the Gauss method to the matrix
\begin{eqnarray}
\left( 
\begin{array}{cccc}
\frac 1{4\pi} (k^2\mathbb I-\pmb k \pmb k^t) &  ig(\omega \mathbb I-\pmb v \pmb k^t) & -i\pmb k & \pmb 0\\
-ig(\omega \mathbb I -\pmb k \pmb v^t) & \frac 1\chi \left(\frac {\omega^2}{\omega_0^2} -1 \right) \mathbb I & \pmb 0 & \pmb v \\
i\pmb k^t & \pmb 0^t &  -\xi & 0  \\
\pmb 0^t & \pmb v^t & 0 & 0
\end{array}
\right|
\left. 
\begin{array}{cccc}
\mathbb I & \mathbb O & \pmb 0 & \pmb 0 \\
\mathbb O & \mathbb I & \pmb 0 & \pmb 0 \\
\pmb 0^t & \pmb 0^t & 1 & 0 \\
\pmb 0^t & \pmb 0^t & 0 & 1 
\end{array}
\right).
\end{eqnarray}
First we multiply the third line by $-1/\xi$, then we add it, multiplied by the column  $i\pmb k $ from the left, to the first line and multiply the second line by $\chi\omega_0^2(\omega^2-\omega_0^2)^{-1}$ to get
\begin{eqnarray}
\left( 
\begin{array}{cccc}
\frac 1{4\pi} \left(k^2\mathbb I-\left(1-\frac {4\pi}\xi \right)\pmb k \pmb k^t\right) &  ig(\omega \mathbb I-\pmb v \pmb k^t) & \pmb 0 & \pmb 0\\
-ig\frac {\chi \omega_0^2}{(\omega^2-\omega_0^2)}(\omega \mathbb I -\pmb k \pmb v^t) & \mathbb I & \pmb 0 & \frac {\chi \omega_0^2}{(\omega^2-\omega_0^2)}\pmb v \\
-\frac i\xi \pmb k^t & \pmb 0^t &  1 & 0  \\
\pmb 0^t & \pmb v^t & 0 & 0
\end{array}
\right|
\left. 
\begin{array}{cccc}
\mathbb I & \mathbb O & -\frac i\xi \pmb k & \pmb 0 \\
\mathbb O & \frac {\chi \omega_0^2}{(\omega^2-\omega_0^2)} \mathbb I & \pmb 0 & \pmb 0 \\
\pmb 0^t & \pmb 0^t & -\frac 1\xi & 0 \\
\pmb 0^t & \pmb 0^t & 0 & 1 
\end{array}
\right).
\end{eqnarray}
Next, to the first line we subtract the second one multiplied by $ig(\omega \mathbb I-\pmb v \pmb k^t)$ from the left, and we get
\begin{eqnarray}
\left( 
\begin{array}{cccc}
M &  \mathbb O & \pmb 0 & \pmb 0\\
-ig\frac {\chi \omega_0^2}{(\omega^2-\omega_0^2)}(\omega \mathbb I -\pmb k \pmb v^t) & \mathbb I & \pmb 0 & \frac {\chi \omega_0^2}{(\omega^2-\omega_0^2)}\pmb v \\
-\frac i\xi \pmb k^t & \pmb 0^t &  1 & 0  \\
\pmb 0^t & \pmb v^t & 0 & 0
\end{array}
\right|
\left. 
\begin{array}{cccc}
\mathbb I & -ig\frac {\chi \omega_0^2}{(\omega^2-\omega_0^2)}(\omega \mathbb I -\pmb v \pmb k^t) & -\frac i\xi \pmb k & \pmb 0 \\
\mathbb O & \frac {\chi \omega_0^2}{(\omega^2-\omega_0^2)} \mathbb I & \pmb 0 & \pmb 0 \\
\pmb 0^t & \pmb 0^t & -\frac 1\xi & 0 \\
\pmb 0^t & \pmb 0^t & 0 & 1 
\end{array}
\right), \label{passaggio}
\end{eqnarray}
with $M$ as in (\ref{MM}), with $\varepsilon=0$. At this point we need to compute $M^{-1}$. Since $M$ is a span of the $4\times 4$ identity $\mathbb I$ and the symmetric tensors of rank two generated by $\pmb k$ and $\pmb v$, the same must happen
for $M^{-1}$, so we look for it as a matrix of the form
\begin{eqnarray}
M^{-1}=\alpha_0 \mathbb I+\alpha_1 \pmb k \pmb k^t+\alpha_2 (\pmb k \pmb v^t +\pmb v \pmb k^t) +\alpha_3 \pmb v \pmb v^t.
\end{eqnarray}
From this, by imposing $M^{-1} M=\mathbb I$ we get (\ref{Mdritta-1}). Then, we first multiply the first line of (\ref{passaggio}) by $M^{-1}$, and next we add it to the third line after multiplication by $i\pmb k^t/\xi$ from the left:
\begin{eqnarray}
\left( 
\begin{array}{cccc}
\mathbb I &  \mathbb O & \pmb 0 & \pmb 0\\
-ig\frac {\chi \omega_0^2}{(\omega^2-\omega_0^2)}(\omega \mathbb I -\pmb k \pmb v^t) & \mathbb I & \pmb 0 & \frac {\chi \omega_0^2}{(\omega^2-\omega_0^2)}\pmb v \\
\pmb 0^t & \pmb 0^t &  1 & 0  \\
\pmb 0^t & \pmb v^t & 0 & 0
\end{array}
\right|
\left. 
\begin{array}{cccc}
M^{-1} & N & -i\frac {\pmb k}{k^2} & \pmb 0 \\
\mathbb O & \frac {\chi \omega_0^2}{(\omega^2-\omega_0^2)} \mathbb I & \pmb 0 & \pmb 0 \\
i\frac {\pmb k^t}{k^2}& \pmb 0^t & 0 & 0 \\
\pmb 0^t & \pmb 0^t & 0 & 1 
\end{array}
\right),
\end{eqnarray}
where $N$ is given in (\ref{N}) and we used that $\pmb k^t N=0$. To the second line we add the first one multiplied by $ig\frac {\chi \omega_0^2}{(\omega^2-\omega_0^2)}(\omega \mathbb I -\pmb k \pmb v^t) $ from the left, and then we subtract the 
second line multiplied by $\pmb v^t$ from the left, to the fourth line, thus getting
\begin{eqnarray}
\left( 
\begin{array}{cccc}
\mathbb I &  \mathbb O & \pmb 0 & \pmb 0\\
\mathbb O & \mathbb I & \pmb 0 & \frac {\chi \omega_0^2}{(\omega^2-\omega_0^2)}\pmb v \\
\pmb 0^t & \pmb 0^t &  1 & 0  \\
\pmb 0^t & \pmb 0^t & 0 & -\frac {\chi \omega_0^2}{(\omega^2-\omega_0^2)}
\end{array}
\right|
\left. 
\begin{array}{cccc}
M^{-1} & N & -i\frac {\pmb k}{k^2} & \pmb 0 \\
N^\dagger & Q+\frac {\chi \omega_0^2}{(\omega^2-\omega_0^2)} \mathbb I & \pmb 0 & \pmb 0 \\
i\frac {\pmb k^t}{k^2}& \pmb 0^t & 0 & 0 \\
\pmb 0^t & -\frac {\chi \omega_0^2}{(\omega^2-\omega_0^2)} \pmb v^t & 0 & 1 
\end{array}
\right), 
\end{eqnarray}
with $N^\dagger$ and $Q$ given as in (\ref{N-daga}) and (\ref{Q}) respectively. Finally, we add the last row multiplied by $\pmb v$ from the right to the second line, and next multiply the last row by $-\frac {(\omega^2-\omega_0^2)}{\chi \omega_0^2}$, and
we get the desired result (\ref{tilde M-1}).

\section{On the propagator}\label{app:propagator}
In order to prove proposition \ref{proposition} we need to integrate out the $k^0$ direction in (\ref{propagatore}). This can be done as usual by means of the methods of complex integration. The prescription must ensure that positive square norm states
propagate forward in time. Since the modes have dispersion relations
\begin{eqnarray}
& 0=\omega^2-\bar \omega^2, \\
& 0=\frac {k^2}{4\pi} -\frac {g^2\chi \omega_0^2\omega^2}{\omega^2-\omega_0^2},\\
& 0=k^2,
\end{eqnarray}
which we will call the {\em b-mode}, the {\em transverse modes}, and the {\em free photon modes} respectively, we see that, given our signature
 for the metric, the right prescription for the free photon modes is $k^2\to k^2+i\varepsilon$, whereas
for the b-mode we can equivalently put $\omega^2\to \omega^2+i\varepsilon$ or $\omega_0^2\to \omega_0^2-i\varepsilon$. However, these two choices are not equivalent for the transverse modes and we now show how the latter choice is the right one.
\begin{lemma}
The right prescription for the correct propagation of all modes is $k^2\to k^2+i\varepsilon$ and $\omega_0^2\to \omega_0^2-i\varepsilon$.
\end{lemma}
\begin{proof}
In order to prove the lemma, let us notice that the propagation in time is provided by the phase factor 
\begin{eqnarray}
e^{-i\pmb k\cdot (\pmb x-\pmb y)}=e^{-ik^0(x^0-y^0)} e^{i\vec k\cdot (\vec x-\vec y)}.
\end{eqnarray}
For $x^0>y^0$, the $k^0$ path must be closed with negative imaginary part in order to apply correctly the residue theorem. So, the poles in $k^0$ with negative imaginary part will contribute to the integral. This means that are just the poles 
corresponding to positive norm states that must have negative imaginary part. This justifies the $i\varepsilon$ prescription for the free photon modes and for the b-mode. For the transverse modes, since in this case the positive norm states 
correspond to positive values of $DR'_{(a)}$, we must check that the solutions of 
\begin{eqnarray}
0=\frac {k^2+i\varepsilon}{4\pi} -\frac {g^2\chi (\omega_0^2-i\varepsilon)\omega^2}{\omega^2-\omega_0^2+i\varepsilon}
\end{eqnarray}
with negative imaginary part correspond exactly to the solution with positive $DR'_{(a)}$. To this end we write the equation in the form
\begin{eqnarray}
(k^2+i\varepsilon)(\omega^2-\omega_0^2+i\varepsilon)=4\pi g^2 \chi (\omega_0^2-i\varepsilon) \omega^2
\end{eqnarray}
and set $k^0_{\varepsilon}=k^0+ig(\vec k)\varepsilon+o(\varepsilon)$. After substitution we get immediately 
\begin{eqnarray}
g(\vec k)=-\frac {1}{2DR'_{(a)}(\vec k)}\left[ 1+\frac {4\pi g^2 \chi \omega_{(a)}^4}{(\omega_{(a)}^2-\omega_0^2)^2} \right],
\end{eqnarray}
which proves the lemma.
\end{proof}
Now we can proceed with the proof of the proposition. Following the notations of section \ref{sec:path integral}, the matrix $\mathcal M^{-1}_{i\varepsilon}$ is written as a $3\times 3$ block matrix, see (\ref{tilde M-1 reduc}). In this way we can
separate the proof in the following steps.
\subsection{(\ref{AA}).}
We will prove that
\begin{eqnarray}
G^{(\mu+1)(\nu+1)}(\pmb x,\pmb y)=\int_{\mathbb R^4} \frac {d^4\pmb k}{(2\pi)^4} {\mathcal M_{i\varepsilon}^{-1}(\pmb k)}^{1,1}e^{-i\pmb k \cdot (\pmb x-\pmb y)}=
\int_{\mathbb R^4} \frac {d^4\pmb k}{(2\pi)^4} {M_{i\varepsilon}^{-1}(\pmb k)} e^{-i\pmb k \cdot (\pmb x-\pmb y)},
\end{eqnarray}
where $1,1$ indicates the first ($4\times 4$) block. We need only to prove it for the inward propagation, that is when closing clockwise the path (when $x^0-y^0>0$). There are three kinds of contributions to the residua.
\subsubsection{The $\pmb{k^0=k^0_>}$ pole.} This gives the contribution to the b-mode. Looking at $M^{-1}_{i\varepsilon}$, we see that the polar part in $k^0_>$ is 
\begin{align}
&\left( -\frac {1}{(k^2)^2} \frac {4\pi (\omega^2-\omega_0^2)}{\omega^2-\bar \omega^2+i\varepsilon} +\frac {4\pi g^2 \chi \omega_0^2}{\left[ \frac {k^2}{4\pi}- \frac {g^2 \chi \omega_0^2 \omega^2}{\omega^2-\omega_0^2} \right] k^2}
\frac 1{\omega^2-\bar \omega^2+i\varepsilon} \right) k^\mu k^\nu\cr 
&-\frac {4\pi g^2 \chi \omega_0^2}{\left[ \frac {k^2}{4\pi}- \frac {g^2 \chi \omega_0^2 \omega^2}{\omega^2-\omega_0^2} \right] (\omega^2-\bar \omega^2+i\varepsilon)}
\left[ \frac \omega{k^2} (k^\mu v^\nu+k^\nu v^\mu) -v^\mu v^\nu \right].
\end{align}
In order to apply the residuum theorem, we note that for $\omega=\bar \omega$ we have
\begin{eqnarray}
\frac {k^2}{4\pi}- \frac {g^2 \chi \omega_0^2 \omega^2}{\omega^2-\omega_0^2}=\frac 1{4\pi} \left( k^2-\bar \omega^2 \right) 
\end{eqnarray}
which substituted above and remembering a $-2\pi i$ factor, reproduces exactly the fourth row of (\ref{AA}).
\subsubsection{The $\pmb{k^0=k^0_{(a)}}$ poles.} These give the contributions to the transverse modes. Near the pole $k^0=k^0_{(a)}$ the matrix $M^{-1}$ is nearly
\begin{align}
M^{-1}_{i\varepsilon}\simeq \frac 1{k^0-k^0_{(a)}+i\varepsilon} \frac 1{DR'_{(a)}(\vec k)} & \left[ 
\eta^{\mu\nu} +\frac {4\pi g^2 \chi \omega_0^2}{k_{(a)}^2 (\omega_{(a)}^2-\omega_0^2)} k_{(a)}^\mu k_{(a)}^\nu \right. \cr
&\left. -\frac {4\pi g^2 \chi \omega_0^2}{\omega_{(a)}^2-\omega_0^2}
\left( \frac {\omega_{(a)}}{k_{(a)}^2} (k_{(a)}^\mu v^\nu +k_{(a)}^\nu v^\mu) -v^\mu v^\nu \right)
\right].
\end{align}
By using that 
\begin{eqnarray}
&& \frac {\omega_{(a)}^2}{k_{(a)}^2}=\frac {\omega_{(a)}^2-\omega_0^2}{4\pi g^2 \chi \omega_0^2},\\
&& \frac {\omega_{(a)}^2}{k_{(a)}^2 }-1=\frac {\omega_{(a)}^2-\bar \omega^2}{4\pi g^2 \chi \omega_0^2}, 
\end{eqnarray}
we see that contracting with $v_\nu$ or $k_\nu$ we get zero, and, being $\pmb e_i^{(a)}$ spacelike and orthogonal to $\pmb v$ and to $\pmb k$, we get that
\begin{eqnarray}
M^{-1}_{i\varepsilon}\simeq -\sum_{i=1}^2 \frac {e^{(a)\mu}_i (\vec k) e^{(a)\nu}_i (\vec k)}{DR'_{(a)}(\vec k)} \frac 1{k^0-k^0_{(a)}+i\varepsilon}.
\end{eqnarray}
Summing up the contributions of both the poles $a=1,2$, and taking into account the factor $-2\pi i$ of the residua theorem, we get the last row of (\ref{AA}).
\subsubsection{The $\pmb{k^0=k^0_+}$ poles.} These give the contributions to the free photon modes. Since we have a second order pole, it is convenient to include the exponential factor in the polar part that is
\begin{align}
e^{-i\pmb k\cdot (\pmb x-\pmb y)}&\left\{ 
\left[
\frac {4\pi}{(k^2+i\varepsilon)^2} \left( \frac\xi{4\pi}- \frac {\omega^2-\omega_0^2}{\omega^2-\bar \omega^2} \right)
+\frac 1{ \frac {k^2}{4\pi}- \frac {g^2 \chi \omega_0^2 \omega^2}{\omega^2-\omega_0^2}}  \frac {g^2\chi \omega_0^2}{\omega^2-\bar\omega^2} \frac {4\pi}{k^2+i\varepsilon}
\right] k^\mu k^\nu
\right.\cr
&\left.
-\frac {g^2\chi \omega_0^2 \omega}{ \frac {k^2}{4\pi}- \frac {g^2 \chi \omega_0^2 \omega^2}{\omega^2-\omega_0^2}}  \frac {k^\mu v^\nu+k^\nu v^\mu}{\omega^2-\bar\omega^2}
\frac {4\pi}{k^2+i\varepsilon}
\right\}.
\end{align}
The residuum is thus
\begin{align}
&\frac 1{2\omega_+ k^0_+} \frac {e^{-i\pmb k_+\cdot(\pmb x-\pmb y)}}{\frac 1{4\pi}-\frac {g^2\chi\omega_0^2}{\omega^2-\omega_0^2}}(k^\mu v^\nu+k^\nu v^\mu-k^\mu k^\nu) 
+\frac d{dk^0} \left[ \frac {e^{-i\pmb k_+\cdot(\pmb x-\pmb y)}}{(k^0+|\vec k|)^2} \left( \xi- \frac 1{\frac 1{4\pi}-\frac {g^2\chi\omega_0^2}{\omega^2-\omega_0^2}} \right) k^\mu k^\nu \right]_{k^0=|\vec k|}\cr
&=\frac 1{2\omega_+ k^0_+} \frac {e^{-i\pmb k_+\cdot(\pmb x-\pmb y)}}{\frac 1{4\pi}-\frac {g^2\chi\omega_0^2}{\omega^2-\omega_0^2}}(k^\mu v^\nu+k^\nu v^\mu-k^\mu k^\nu) 
-i(x^0-y^0) \frac {e^{-i\pmb k_+\cdot(\pmb x-\pmb y)}}{4|\vec k|^2}  \left( \xi- \frac 1{\frac 1{4\pi}-\frac {g^2\chi\omega_0^2}{\omega^2-\omega_0^2}} \right) k^\mu k^\nu \cr
&\quad -\frac {e^{-i\pmb k_+\cdot(\pmb x-\pmb y)}}{4|\vec k|^3}  \left( \xi- \frac 1{\frac 1{4\pi}-\frac {g^2\chi\omega_0^2}{\omega^2-\omega_0^2}} \right) k^\mu k^\nu 
+\frac {e^{-i\pmb k_+\cdot(\pmb x-\pmb y)}}{4|\vec k|^2}  \left( \xi- \frac 1{\frac 1{4\pi}-\frac {g^2\chi\omega_0^2}{\omega^2-\omega_0^2}} \right) (\eta^{0\mu} k_+^\nu+\eta^{0\nu} k_+^\mu)\cr
& +\frac {e^{-i\pmb k_+\cdot(\pmb x-\pmb y)}}{4|\vec k|^2} \frac {2\omega_+ v^0 g^2\chi \omega_0^2}{(\omega_+^2-\omega_0)^2} \frac {k^\mu k^\nu}{\left[ \frac 1{4\pi}-\frac {g^2\chi\omega_0^2}{\omega^2-\omega_0^2} \right]^2}.
\end{align} 
Apparently, this expression does not reproduce the first three rows of (\ref{AA}). However, it is easy to see that these are reproduced integrating in the direction $k^\mu v_\mu$ in place of the direction $k^0$. This can be done by taking
a boost such that $v^\mu\to (1,\vec 0)$, integrating in the new $k^0$ direction and then going back to the original frame. This completes the proof of the first statement.
\subsection{(\ref{AP}).}
We will prove that
\begin{eqnarray}
G^{(\mu+1)(\nu+5)}(\pmb x,\pmb y)=\int_{\mathbb R^4} \frac {d^4\pmb k}{(2\pi)^4} {\mathcal M_{i\varepsilon}^{-1}(\pmb k)}^{1,2}e^{-i\pmb k \cdot (\pmb x-\pmb y)}=\int_{\mathbb R^4} \frac {d^4\pmb k}{(2\pi)^4} {N_{i\varepsilon}(\pmb k)}e^{-i\pmb k \cdot (\pmb x-\pmb y)}.
\end{eqnarray}
We need only to prove it for the inward propagation, that is when closing clockwise the path (when $x^0-y^0>0$). There are three kinds of contributions to the residua.
\subsubsection{The $\pmb{k^0=k^0_>}$ pole.} This gives the contribution to the b-mode. Looking at $N_{i\varepsilon}$, we see that the polar part in $k^0_>$ is 
\begin{align}
N_{i\varepsilon}\simeq -ig{2v^0} \frac {4\pi \chi \omega_0^2}{k_>^2-\bar \omega^2} \left( v^\mu v^\nu +\frac {k_>^\mu k_>^\nu}{k_>^2}-\frac {\bar \omega}{k^2_>} v^\mu k_>^\nu -\frac {k_>^\mu v^\nu}{\bar\omega} \right)
\frac 1{k^0-k^0_>+i\varepsilon},
\end{align}
which leads immediately to the second line of (\ref{AP}). 
\subsubsection{The $\pmb{k^0=k^0_{(a)}}$ poles.} These give the contributions to the transverse modes. Near the pole $k^0=k^0_{(a)}$ the matrix $N_{i\varepsilon}$ is nearly
\begin{align}
N_{i\varepsilon}\simeq - \frac {ig\chi \omega_0^2}{\omega_{(a)}^2-\omega_0^2} \frac {\omega_{(a)}}{DR'_{(a)}(\vec k) (k^0-k^0_{(a)}+i\varepsilon)} &\left[ \eta^{\mu\nu} +\frac {4\pi g^2 \chi \omega_0^2}{\omega_{(a)}^2-\omega_0^2}
\left( v^\mu v^\nu \right.\right. \cr
& \left.\left. +\frac {k_{(a)}^\mu k_{(a)}^\nu}{k_{(a)}^2} -\frac {\omega_{(a)}}{k^2_{(a)}} v^\mu k^\nu -\frac {\omega_{(a)}^2-\omega_0^2}{4\pi\omega_{(a)}g^2\chi \omega_0^2} k^\mu v^\nu\right) \right],
\end{align}
and using that 
\begin{eqnarray}
&& \frac {\omega_{(a)}^2}{k_{(a)}^2}=\frac {\omega_{(a)}^2-\omega_0^2}{4\pi g^2 \chi \omega_0^2},
\end{eqnarray}
we see that contracting with $v_\nu$ or $k_\nu$ we get zero. Since $\pmb e_i^{(a)}$ spacelike and orthogonal to $\pmb v$ and to $\pmb k$, we get that
\begin{eqnarray}
N_{i\varepsilon}\simeq ig\chi \omega_0^2 \sum_{i=1}^2 \frac {e^{(a)\mu}_i (\vec k) e^{(a)\nu}_i (\vec k)}{DR'_{(a)}(\vec k)} \frac {\omega_{(a)}}{\omega_{(a)}^2-\omega_0^2}\frac 1{k^0-k^0_{(a)}+i\varepsilon}.
\end{eqnarray}
Summing up the contributions of both the poles $a=1,2$, and applying the theorem of residues, we get the last row of (\ref{AP}).
\subsubsection{The $\pmb{k^0=k^0_+}$ poles.} These give the contributions to the free photon modes. We see that for $k^0\simeq k^0_+$
\begin{eqnarray}
N_{i\varepsilon}\simeq \frac i{g\omega_+} \frac {4\pi g^2 \chi \omega_0^2}{\omega_+^2-\omega_0^2} \frac 1{2|\vec k|} \left( k^\mu k^\nu-\omega_+ v^\mu k^\nu \right) \frac 1{k^0-k^0_++i\varepsilon}.
\end{eqnarray}
This leads immediately to the first row of (\ref{AP}).
\subsection{(\ref{AB}).}
We will prove that
\begin{eqnarray}
G^{(\mu+1)(9)}(\pmb x,\pmb y)=-i\int_{\mathbb R^4} \frac {d^4\pmb k}{(2\pi)^4} \frac {k^\mu}{k^2+i\varepsilon} e^{-i\pmb k \cdot (\pmb x-\pmb y)}.
\end{eqnarray}
We need only to prove it for the inward propagation, that is when closing clockwise the path (when $x^0-y^0>0$). In this case there is just one contribution, which corresponds to the pole $k^0=k^0_+-i\varepsilon$. Here, the direct computation
gives the right result.
\subsection{(\ref{PP}).}
We will prove that
\begin{eqnarray}
G^{(\mu+5)(\nu+5)}(\pmb x,\pmb y)=\int_{\mathbb R^4} \frac {d^4\pmb k}{(2\pi)^4} {\mathcal M_{i\varepsilon}^{-1}(\pmb k)}^{2,2}e^{-i\pmb k \cdot (\pmb x-\pmb y)}=\int_{\mathbb R^4} \frac {d^4\pmb k}{(2\pi)^4}  \left(Q^{\mu\nu}_{i\varepsilon}(\pmb k)
+\frac {\chi \omega_0^2}{\omega^2-\omega_0^2+i\varepsilon} \eta^{\mu\nu} \right)e^{-i\pmb k \cdot (\pmb x-\pmb y)}.
\end{eqnarray}
We need only to prove it for the inward propagation, that is when closing clockwise the path (when $x^0-y^0>0$). It is interesting to note that in this case we have three contribution, but the pole in $k^0_+$ is replaced by a pole
in $\omega=\omega_0$. This contribution corresponds to the solution $\pmb A=\pmb 0$, $B=0$, $\pmb P\propto \pmb v$, which must be discarded because of the condition $\pmb v \cdot \pmb P=0$.
\subsubsection{The $\pmb{k^0=k^0_>}$ pole.}
In this case
\begin{align}
Q_{i\varepsilon}\simeq \frac {\chi \omega_0^2}{2\bar \omega v^0} \frac 1{k_>^2-\bar \omega^2} \left( \bar \omega v^\mu- k_>^\mu \right)\left( \bar \omega v^\nu- k_>^\nu \right) \frac 1{k^0-k^0_>+i\varepsilon},
\end{align}
which, through the theorem of residues, leads to the first row of (\ref{PP}).
\subsubsection{The $\pmb{k^0=k^0_{(a)}}$ poles.}
Again, the polar part is in $Q_{i\varepsilon}$ only, and, as in the previous subsections, it is sufficient to check that $\pmb k$ and $\pmb v$ are in the kernel of the polar part of $Q_{i\varepsilon}$. But this is easily checked exactly in the same way as for the previous 
subsections.

\

\

\

\

\noindent Thus, we are left with the expressions (\ref{PB}) and (\ref{BB}), which, however, are trivially verified. Then, the proof of the proposition is complete.

\end{appendix}


\end{document}